\documentclass[12pt]{article}
\usepackage[english]{babel}
\usepackage[english]{babel}
\usepackage{amssymb}
\usepackage{amsmath}
\usepackage{amsfonts}
\usepackage{amsbsy}
\usepackage{indentfirst}
\usepackage{graphicx}
\usepackage{color}

\newtheorem{teor}{Theorem}[section]
\newtheorem{pro}{Proposition}[section]

\newtheorem{lema}{Lemma}[section]

\newenvironment{proof}[1][Proof]{\noindent\textbf{#1.} }{\ \rule{0.5em}{0.5em}}
\begin{document}
\newcommand{\bfA}{\mathbf{A}}
\newcommand{\C}{\mathbb{C}}
\newcommand{\R}{\mathbb{R}}
\newcommand{\N}{\mathbb{N}}
\newcommand{\bfgr}{\boldsymbol{\nabla}}
\newcommand{\bfal}{\boldsymbol{\alpha}}
\newcommand{\bfpi}{\boldsymbol{\pi}}
\newcommand{\bfta}{\boldsymbol{\tau}}
\newcommand{\bfr}{\mathbf{r}}
\newcommand{\bfq}{\mathbf{q}}
\newcommand{\e}{\mathrm{e}}
\def\d{{\rm d}}
\def\tr{{\rm tr}}
\def\Tr{{\rm Tr}}
\def\hS{{\hat S}}
\def\al{\alpha}
\def\1{{\bf 1}}
\def\g{\gamma}
\font\titlefont=cmbx10 scaled\magstep1
%%%%%%%%%%%%%%%%%%%%%%%%%%%%%%%%%%%%%%%%5

\title{Fermionic model with a non-Hermitian Hamiltonian}
\author{N. Bebiano\footnote{ CMUC, University of Coimbra, Department of
Mathematics, P 3001-454 Coimbra, Portugal (bebiano@mat.uc.pt)},
J.~da Provid\^encia$^\dag$, S. Nishiyama\footnote{University of Coimbra, CFisUC, Department of
Physics, P 3004-516 Coimbra, Portugal (providencia@teor.fis.uc.pt)}~
%S.~Nishiyama\footnote{University of Coimbra, CFisUc, Department of
%Physics, P 3004-516 Coimbra, Portugal (seikoceu@khe.biglobe.ne.jp)}~
and J.P.
da Provid\^encia\footnote{Depatamento de F\'\i sica, University of Beira
Interior, P-6201-001 Covilh\~a, Portugal
(joaodaprovidencia@daad-alumni.de)}
}
%\date{}
\maketitle
\begin{abstract}
This paper deals with the mathematical spectral analysis and physical interpretation
of a fermionic system described by a non-Hermitian  Hamiltonian possessing real eigenvalues.
A statistical thermodynamical description of such a system is considered. Approximate expressions for the energy expectation value
and the number operator expectation value, in terms of the absolute temperature $T$ and of the chemical potential $\mu$, are obtained, based
on the  Euler-Maclaurin formula.
% The thermodynamical inequality for such systems
%%described by a non-Hemitian Hamiltonian with real eigenvalues
%is formulated.
%%%%%%%%%%%%%%
\end{abstract}
%\section{Introduction}
%This problem deals with the mathematical description and physical interpretation
%of a fermionic system described by a non-Hermitian  Hamiltonian possessing real eigenvalues.
%%%%%%%%%%%%%%

\noindent
AMS subject classification: 15A75

\noindent
Key words:
antisymmetric tensor,  Hamiltonian, fermion, pseudo-fermion,  Complete biorthogonal
sets, Hermitian, non-Hermitian, metric matrix,
partition function
%and useful mathematical background in this context is pointed out.
\section{Introduction: fermionic creation and annihilation operators}\label{S1}
%which naturally avoids many difficulties inherent to the theories of bosonic systems.
%Fermionic operators obey {\it cannonical anticommutation relations}
%$$\{c,c^\dag \}=cc^\dag+c^\dag c=\1,\quad c^2=c^{\dag~2}=0, $$
%where $c^\dag$ is the {\it adjoint} of $c$ and $\1$ is the identity operator.

The {\it observables} of a physical system are usually Hermitian operators, which describe measurable quantities.
Elementary particles are either {\it bosons} or {\it fermions}.
Fermions, contrarily to bosons,
are described by bounded operators.
%%%%%%%%%%%%%%%%%%%%%%%
%In order to introduce the Hilbert space where the observables of our model act, some
%preliminaries are in order.
%%%%%%%%%%%%
The fermionic operators act on an infinite-dimensional Hilbert space ${\cal H}$.
%endowed
More concretely,
$\cal H$ is the direct sum of the spaces of {\it completely antisymmetric tensors} of rank $k$,
$k=0,1,2,3, \ldots $ over $\C^\infty$,
$${\cal H}={\cal A}_0\oplus{\cal A}_1\oplus{\cal A}_2\oplus\ldots%{\cal A}^4_3\oplus{\cal A}^n_n
,$$
with the inner product $\langle\cdot,\cdot\rangle$ defined below.
%while $\cal H$ is the Hilbert space of the
%fermionic system, constituted by $0,1,2,\ldots$ particles.
The Hilbert space of individual fermionic states is $\C^\infty$.
Let $\{e_j\}$ be an orthonormal (o.n.) basis of $\C^\infty$.
%$^n_k$%and let ${\cal A}^\infty_k$ be the space of absolutely {\it antisymmetric tensors} of rank $k$ over $\C^\infty$.
By $\{e_{j_1}\wedge\cdots\wedge e_ {j_k}:j_1,\ldots,j_k=1,2,3,\ldots\}$, we denote a basis of ${\cal A}_k$,
constituted by the following  tensors
$$e_{j_1}\wedge\cdots\wedge e_ {j_k}=\frac{1}{\sqrt{k!}}\sum_{\sigma\in S_k}{\rm sign}(\sigma)e_{j_{\sigma(1)}}\otimes\cdots\otimes e_{j_{\sigma(k)}},$$
where
$S_k$ is the {\it symmetric group of degree} $k$ and ``sign" represents the $\pm$ sign of the permutation.
Clearly
$$e_{j_1}\wedge\cdots\wedge e_ {j_k}={\rm sign}(\sigma)e_{j_{\sigma(1)}}\wedge\cdots\wedge e_{j_{\sigma(k)}}. $$
%%%%%%%%%%%%%%
%A general tensor in ${\cal A}^\infty_k$ is written as
%$$\psi=\sum_{1\leq j_1<\ldots<j_k\leq \infty}\psi_{j_1,\ldots,j_k}e_{j_1}\wedge\cdots\wedge e_ {j_k},\quad \psi_{j_1,\ldots,j_k}\in \C.$$

The inner product  in ${\cal A}_k$ is defined as
\begin{eqnarray}\langle\phi,\psi\rangle_k:=\sum_{1\leq j_1<\ldots<j_k\leq \infty}\overline{\psi_{j_1,\ldots,j_k}}\phi_{j_1,\ldots,j_k},\label{lrk}\end{eqnarray}
where,
%for ${\psi_{j_1,\ldots,j_k}},~\phi_{j_1,\ldots,j_k}\in\C$,
\begin{eqnarray*}
&&\phi=\sum_{1\leq j_1<\ldots<j_k\leq \infty}\phi_{j_1,\ldots,j_k}e_{j_1}\wedge\cdots\wedge e_ {j_k},\quad
%A general tensor in ${\cal A}^n_k$ is written as
\psi=\sum_{1\leq j_1<\ldots<j_k\leq \infty}\psi_{j_1,\ldots,j_k}e_{j_1}\wedge\cdots\wedge e_ {j_k}.
\end{eqnarray*}
%%%%%%%%%%%%%%%%%
%A general tensor in ${\cal A}^\infty_k$ is written as
%$$\psi=\sum_{1\leq j_1<\ldots<j_k\leq \infty}\psi_{j_1,\ldots,j_k}e_{j_1}\wedge\cdots\wedge e_ {j_k},\quad \psi_{j_1,\ldots,j_k}\in \C.$$

The inner product  in ${\cal H}$ is defined as
\begin{eqnarray}\langle\phi,\psi\rangle_{\cal H}:=\overline{\psi_{j_0}}\phi_{j_0}+\sum_{k=1}^\infty~\sum_{1\leq j_1<\ldots<j_k\leq \infty}\overline{\psi_{j_1,\ldots,j_k}}\phi_{j_1,\ldots,j_k},~~{\psi_{j_0}},~\phi_{j_0}\in\C,\label{lrcalH}\end{eqnarray}
where, for ${\psi_{j_1,\ldots,j_k}},~\phi_{j_1,\ldots,j_k}\in\C,$
\begin{eqnarray*}
&&\phi=\phi_{j_0}+\sum_{k=1}^\infty~\sum_{1\leq j_1<\ldots<j_k\leq \infty}\phi_{j_1,\ldots,j_k}e_{j_1}\wedge\cdots\wedge e_ {j_k},\\&&
%A general tensor in ${\cal A}^n_k$ is written as
\psi=\psi_{j_0}+\sum_{k=1}^\infty~\sum_{1\leq j_1<\ldots<j_k\leq \infty}\psi_{j_1,\ldots,j_k}e_{j_1}\wedge\cdots\wedge e_ {j_k}.
\end{eqnarray*}
%%%%%%%%%%%%%%%%%
%%%%%%%%%%%%%%%%%%%%%%%

We consider {\it fermionic creation operators} $c_j^\dag:{\cal A}_{k}\rightarrow{\cal A}_{k+1}$,
$$
c_j^\dag ~e_{j_1}\wedge\cdots\wedge e_{j_k}=e_j\wedge e_{j_1}\wedge\cdots\wedge e_{j_k},$$
and {\it fermionic annihilation operators} $c_j:{\cal A}_{k}\rightarrow{\cal A}_{k-1}$,
\begin{eqnarray*}
&&
c_j\psi=0\quad{\rm if}\quad\psi\in{\cal A}_0, \quad c_{j_1}~ e_{j_1}\wedge\cdots\wedge e_{j_k}= e_{j_2}\wedge\cdots\wedge e_{j_k},\\
&&
{\rm and}\quad c_{j}~ e_{j_1}\wedge\cdots\wedge e_{j_k}= 0\quad{\rm if}\quad j\notin\{j_1,\ldots j_k\}.
\end{eqnarray*}
The operator $c^\dag_j$ is the {\it adjoint} of $c_j$,
$$\langle c^\dag_j\phi,\psi\rangle=\langle\phi,c_j\psi\rangle,\quad\phi,\psi\in {\cal H}.$$
%%%%%%%%%%%%%%%%%%%%%
The following {\it anticommutation relations} for the fermionic operators $c_j^\dag,c_j$ hold:
$$\{c_i^\dag, c_j\}=c^\dag_i c_j+c_jc^\dag_i=\delta_{ij},~~\{c_i^\dag, c^\dag_j\}=\{c_i, c_j\}=0,~~i,j=1,2,3,\ldots,n,$$
where $\delta_{ij}$ is the {\it Kronecker symbol}. Notice that
$$[c^\dag_ic_j,c^\dag_k]=-\delta_{jk}c^\dag_i, ~~[c^\dag_ic_j,c_k]=\delta_{ik}c_j.$$
The {\it number operator} in state $i$ is the Hermitian operator given by
$${N_{op}}_i=c_i^\dag c_i,$$
and its eigenvalues are 0 and 1, being the number of fermions in that state. The {\it total number operator} is $N_{op}={N_{op}}_1+{N_{op}}_2+{N_{op}}_3\ldots.$

In the last two decades, the quantum physics of systems described by non-Hermitian Hamiltonians has attracted
the attention of researchers, from mathematitians to theoretical and applied physicists. Several classes of models have been investigated including
bosonic systems, relevant in the so-called $PT$- and {\it pseudo-Hermitian}-quantum mechanics,
see \cite{bagarello,bagarello1,bagarello*,bagarelo**,*,bebiano*,bebiano,beb,[1],bender,providencia} and references therein.

In the context of quantum systems with non-Hermitian Hamiltonians,
{\it pseudo fermionic} operators appear, instead of fermionic operators,
and the anticommutation relations are replaced by $\{a,b\}=1$. In this case, $a^2=b^2=0,$ but
$b$ is not assumed to be the {\it adjoint} of $a$.
%%%%%%%%%%%%%%%%%%%%%%%%%%

The paper is organized as follows. In Section \ref{S2}, the spectral analysis  of a non-Hermitian semi-infinite matrix
$\widehat H$ %with real eigenvalues}
is performed.
In Section \ref{S2a}, a metric matrix which renders $\widehat H$ Hermitian is constructed. The obtained results are crucial in the remaining parts of
the paper. In Section \ref{S3}, {a fermionic model} characterized by {a non-Hermitian Hamiltonian with real eigenvalues}
is introduced. In Section \ref{S4}, the fermionic Hamiltonian is expressed in terms of {dynamical pseudo-fermionic operators},
using the results in Section \ref{S2}. In Section \ref{S6}, the so called Physical Hilbert space is introduced.
In Section \ref{S5},
{statistical thermodynamics  considerations are applied to the studied fermionic Hamiltonian.} A numerical Example is also presented.
In Section \ref{S8}, some conclusions are drawn.
\section{A non-Hermitian matrix with real eigenvalues}\label{S2}
Let us consider the semi-infinite tridiagonal matrix, which has a central role in the paper,
\begin{eqnarray}
\widehat H=\frac{1}{2\sqrt2}\left[\begin{matrix}
1/\sqrt2&-\sqrt{{1\cdot2}{}}\gamma&0&0&\ldots\\
\sqrt{{1\cdot2}{}}\gamma&5/\sqrt2&-\sqrt{{3\cdot4}{}}\gamma&0&\ldots\\
0&\sqrt{{3\cdot4}{}}\gamma&9/\sqrt2&-\sqrt{{5\cdot6}{}}\gamma&\ldots\\
0&0&\sqrt{{5\cdot6}{}}\gamma&13/\sqrt2%&-\sqrt{\frac{3\cdot4}{2}}\gamma
&\ldots\\
\vdots&\vdots&\vdots&\vdots&\ddots\end{matrix}\right],~~~\gamma\in\R.\label{widehatH}
\end{eqnarray}
We will be concerned with the following matrices
\begin{eqnarray}&&
\widehat S_+=\frac{1}{2\sqrt2}\left[\begin{matrix}
0&0&0&0&\ldots\\
\sqrt{{1\cdot2}{}}&0&0&0&\ldots\\
0&\sqrt{{3\cdot4}{}}&0&0&\ldots\\
0&0&\sqrt{{5\cdot6}{}}&0%&-\sqrt{\frac{3\cdot4}{2}}\gamma
&\ldots\\
\vdots&\vdots&\vdots&\vdots&\ddots\end{matrix}\right]~\nonumber\\&&
%%%%%%%%%%%%%%%%%%%%%%
\widehat S_-=(\widehat S_+)^T\nonumber\\
&&\widehat S_0={\rm diag}(1/4,5/4,9/4 ,\ldots)\label{hatS+S-S0}.
\end{eqnarray}
The following commutation relations are easily seen to be satisfied
\begin{eqnarray*}
[\widehat S_-,\widehat S_0]=\widehat S_-,\quad
[\widehat S_0,\widehat S_+]=\widehat S_+,\quad
[\widehat S_-,\widehat S_+]=\widehat S_0.\quad
\end{eqnarray*}
Obviously,
$$\widehat H=\widehat S_0+\gamma(\widehat S_+-\widehat S_-).$$

\begin{pro}\label{P2.1}
The spectrum of $\widehat H$ is
%%%%%%%%%%%%%%%%%%%%%%%%%%%%%%%%%%%%%%%%%
$$\sigma(\widehat H)=\sqrt{1+2\g^2} ~\{1/4,5/4,9/4,\ldots\},$$
and the corresponding eigenvectors of $\widehat H$ are given by
$$\widehat \psi_n=\widehat S_+^{n-1}\widehat\psi_1,~~n=1,2 ,\ldots$$
where $\widehat \psi_1$ is such that
$$\widehat S_-\widehat\psi_1=0.$$
\end{pro}
\begin{proof}
In order to obtain the eigenvalues and eigenvectors of $\widehat H$ the {\it equation of motion method} (EMM) \cite{rowe,bebiano*} is used,
\begin{eqnarray*}
&&[\widehat H,z{\widehat S_0}+x{\widehat S_+}+y{\widehat S_-}]={{\widehat S}_0}(-\gamma x-\gamma y)+{\widehat S_+}(-\gamma z+x)+{\widehat S_-}(-\gamma z- y)\\
&&=\Lambda(z{\widehat S_0}+x{\widehat S_+}+y{\widehat S_-}),~~\Lambda,x,y,z\in\R. %,~~k=1,2,3,\ldots.
\end{eqnarray*}
This method leads to the $3\times3$ matrix eigenproblem,
$$
\left[\begin{matrix}0&-\gamma&-\gamma\\
-\gamma&1&0\\
-\gamma&0&-1\end{matrix}\right]
\left[\begin{matrix}z\\x\\y\end{matrix}\right]=\Lambda
\left[\begin{matrix}z\\x\\y\end{matrix}\right],~~\Lambda\in\R.
$$
The eigenvalues are easily obtained,
$$\Lambda_0=0,~\Lambda_1= - \sqrt{1 + 2\g^2},~\Lambda_2=  \sqrt{1 + 2\g^2},$$
as well as the respective eigenvectors
\begin{eqnarray*}
&&u_0=%\frac{1}{\sqrt{1 + 2 \g^2}}
\left[(1, \g, - \g)\right]^T,\\
&&u_-=\left[{1},- \frac{1 - \sqrt{1 + 2 \g^2}}{2\g},
\frac{1+\sqrt{1+2\g^2}}{2\g}\right]^T,\\
&&u_+=\left[{1}, - \frac{1 + \sqrt{1 + 2 \g^2}}{2\g},
\frac{1-\sqrt{1+2\g^2}}{2\g}\right]^T.
\end{eqnarray*}
From the normalized eigenvectors, the following matrices are constructed:
\begin{eqnarray}
&&{\widehat T_0}=\frac{1}{\sqrt{1 + 2 \g^2}}({\widehat S_0}+ \g ({\widehat S_+} - {\widehat S_-})),\nonumber\\
&&{\widehat T_-}=\frac{\g}{{\sqrt{1 + 2 \g^2}}}{\widehat S_0}- \frac{1 - \sqrt{1 + 2 \g^2}}{2 {\sqrt{1 + 2 \g^2}}}{\widehat S_+}
 +\frac{1+\sqrt{1+2\g^2}}{2{\sqrt{1 + 2 \g^2}}}{\widehat S_-},\nonumber\\
 &&{\widehat T_+}=-\frac{\g}{{\sqrt{1 + 2 \g^2}}}{\widehat S_0} + \frac{1 + \sqrt{1 + 2 \g^2}}{2 {\sqrt{1 + 2 \g^2}}}{\widehat S_+}
  -\frac{1-\sqrt{1+2\g^2}}{2{\sqrt{1 + 2 \g^2}}}{\widehat S_-}.\label{hatT0T+T-}
\end{eqnarray}
These matrices
 obey the same commutation relations as the matrices $\widehat S_0,\widehat S_+,\widehat S_-$:
\begin{eqnarray*}
[\widehat T_-,\widehat T_0]=\widehat T_-,\quad
[\widehat T_0,\widehat T_+]=\widehat T_+,\quad
[\widehat T_-,\widehat T_+]=\widehat T_0,\quad
\end{eqnarray*}
and they  characterize the $su(1,1)$ algebra.
%%%%%%%%%%%%%%%%%%%%%%%%%%%%%

In order to determine $\widehat\psi_1$, we notice that
\begin{eqnarray*}
\widehat T_-=\frac{1}{4}\left[\begin{matrix}
1&\sqrt{{1\cdot2}{}}~\eta^{-1}&0&0&\ldots\\
\sqrt{{1\cdot2}{}}~\eta&5&\sqrt{{3\cdot4}{}}~\eta^{-1}&0&\ldots\\
0&\sqrt{{3\cdot4}{}}~\eta&9&\sqrt{{5\cdot6}{}}~\eta^{-1}&\ldots\\
0&0&\sqrt{{5\cdot6}{}}~\eta&13%&-\sqrt{\frac{3\cdot4}{2}}\gamma
&\ldots\\
\vdots&\vdots&\vdots&\vdots&\ddots\end{matrix}\right],~
\end{eqnarray*}
where $$\eta=\frac{\sqrt{1+2\gamma^2}-1}{\sqrt2~\gamma}.$$
We easily  find that
$$\widehat\psi_1=\left[1,~-\sqrt{\frac{1}{2}}~\eta,~\sqrt{\frac{1\cdot3}{2\cdot4}}~\eta^2,
~-\sqrt{\frac{1\cdot3\cdot5}{2\cdot4\cdot6}}~\eta^3,\ldots\right]^T,$$
satisfies
$$\widehat H\widehat\psi_1=\frac{1}{4}\widehat\psi_1.$$
%%%%%%%%%%%%%%
Next, we notice he following:
If $\Lambda$ is an eigenvalue of $\widehat H$ with
eigenvector $\widehat \psi$,
$$\widehat H\widehat \psi=\Lambda\widehat \psi,$$
then $(\Lambda+\sqrt{1+2\gamma^2})$ is an eigenvalue of $\widehat H$ with
eigenvector $\widehat T_+\widehat \psi$, that is,
$$\widehat H\widehat T_+\widehat \psi=\left(\Lambda+\sqrt{1+2\gamma^2}\right)\widehat T_+\widehat \psi.$$
Similarly,
%if $\lambda$ is an eigenvalue of $H$ with eigenvector $\psi$, i.e.,$$H\psi=\lambda\psi,$$then
if $(\Lambda-\sqrt{1+2\gamma^2})$ is an eigenvalue of $\widehat H$ with
eigenvector $\widehat T_-\widehat\psi$,
$$\widehat H\widehat T_-\widehat \psi=\left(\Lambda-\sqrt{1+2\gamma^2}\right)\widehat T_-\widehat \psi,$$
provided $\widehat T_-\widehat \psi\neq0.$
%%%%%%%%%%%%%%%%%%%%%%%%%
%%%%%%%%%%%%%%%%%%%
Now, the claim easily follows.
\end{proof}\\

\begin{pro}\label{P2.2}
%%%%%%%%%%%%%%%%%%%%%%%%%
The eigenvectors of $\widehat H^T$ are given by
$${\breve\psi}_n=(\widehat T_-^T)^{n-1}{\breve\psi}_1,~~n=1,2,\ldots$$
where ${\breve\psi}_1$ is such that
$$(\widehat T_+)^T{\breve\psi}_1=0.$$
\end{pro}
\begin{proof}
We easily find that
$$\breve\psi_0=\left[1,~\sqrt{\frac{1}{2}}~\eta,~\sqrt{\frac{1\cdot3}{2\cdot4}}~\eta^2,
~\sqrt{\frac{1\cdot3\cdot5}{2\cdot4\cdot6}}~\eta^3,\ldots\right]^T,$$
satisfies
$$\widehat H\breve\psi_1=\frac{1}{4}\breve\psi_1.$$
It is now clear that the claim holds.
\end{proof}

Some observations are in order.

%For $|\gamma|<1,$ the eigenvectors $\widehat \psi_0,~{\breve\psi}_0$ are easily obtained, recursively.
1. For a convenient normalization, the eigensystems $\{\widehat \psi_n\},~\{{\breve\psi}_n\}$
become biorthogonal
$$\langle{\breve\psi}_m,\widehat \psi_n\rangle=\delta_{mn}\langle{\breve\psi}_n,\widehat \psi_n\rangle.$$

%%%%%%%%%%%%%%%%%%%%%%%%%%%%%%%
%%%%%%%%%%%%%%%%%%%%%%%%%%%%%%%%%%%%
2.
The matrix $\widehat T_+$ is a {\it raising matrix},
and  $\widehat T_-$ is a  {\it lowering matrix}.
However, %the {\it raising operator}
$\widehat T_+$ is not the adjoint of
$\widehat T_-,~\widehat T_-\neq(\widehat T_+)^\dag$ and $\widehat T_0$ is not Hermitian, $\widehat T_0\neq\widehat T_0^\dag$.
Due to these facts, we  say
 that the matrices $\widehat T_0,\widehat T_+,\widehat T_-$ generate a {\it pseudo}-$su(1,1)$ algebra.
%%%%%%%%%%%%%%%%%
\section{Metric matrix}\label{S2a}
We may easily construct a positive definite matrix  $\widehat D$ and a Hermitian matrix $\widehat H_0$ such that
$$
\widehat H=\widehat D^{-1}\widehat H_0\widehat D,~~
\widehat H^\dag=\widehat D\widehat H_0\widehat D^{-1},
$$
so that
$$\widehat H=\widehat D^{-2}\widehat H^\dag \widehat D^2.$$
%In infinite dimension, such a $\widehat D$ can only be formally introduced.
%Its mathematical existence may be questioned.

For $\widehat S_0,~\widehat S_+,~\widehat S_-$ in (\ref{hatS+S-S0}),
using the commutation relations
$$[\widehat S_0,(\widehat S_++\widehat S_-)]=\widehat S_+-\widehat S_-,~~[(\widehat S_+-\widehat S_-),(\widehat S_++\widehat S_-)]=-2 \widehat S_0,$$
by some calculations, we get
\begin{eqnarray}
&&\e^{-\alpha(\widehat S_++\widehat S_-)}
\widehat S_0\e^{\alpha(\widehat S_++\widehat S_-)}=\cos\sqrt2\alpha~ \widehat S_0+\frac{\widehat S_+-\widehat S_-}{\sqrt2}\sin\sqrt2\alpha,\label{S0*}\\
&&\e^{-\alpha(\widehat S_++\widehat S_-)}
(\widehat S_+-\widehat S_-)\e^{\alpha(\widehat S_++\widehat S_-)}=\cos\sqrt2\alpha~(\widehat S_+-\widehat S_-)-{\widehat S_0}{\sqrt2}\sin\sqrt2\alpha,\nonumber\\
&&\e^{-\alpha(\widehat S_++\widehat S_-)}
(\widehat S_++\widehat S_-)\e^{\alpha(\widehat S_++\widehat S_-)}=\widehat S_++\widehat S_-.\nonumber
%\\&&\e^{-\alpha(\widehat S_++\widehat S_-)}
%(\widehat S_++\widehat S_-)\e^{\alpha(\widehat S_++\widehat S_-)}=\widehat S_++\widehat S_-.\nonumber
\end{eqnarray}
It follows that
\begin{eqnarray}
&&\e^{-\alpha(\widehat S_++\widehat S_-)}\widehat S_+\e^{\alpha (\widehat S_++\widehat S_-)}=\frac{\cos\sqrt2\alpha}{2}~(\widehat S_+-\widehat S_-)
+\frac{1}{2}(\widehat S_++\widehat S_-)-{\widehat S_0}\frac{{\sqrt2}\sin\sqrt2\alpha}{2},\nonumber\\
&&\e^{-\alpha(\widehat S_++\widehat S_-)}\widehat S_-\e^{\alpha (\widehat S_++\widehat S_-)}=\frac{\cos\sqrt2\alpha}{2}~(\widehat S_--\widehat S_+)
+\frac{1}{2}(\widehat S_++\widehat S_-)+{\widehat S_0}\frac{{\sqrt2}\sin\sqrt2\alpha}{2}.\nonumber\\&&\label{S+S-}
\end{eqnarray}
We have shown the following result.
\begin{pro}
Let
%%%%%%%%%%%%%%%%%%%%%%%%%%%%%%%%%%%%%%%55
$$\gamma=\frac{\tan\sqrt2\alpha}{\sqrt2},~~\frac{1}{\sqrt{1+2\gamma^2}}=\cos\sqrt2\alpha,$$
in (\ref{S0*}) and (\ref{S+S-}).
Then, for $\widehat T_0,~\widehat T_+$ and $\widehat T_-$ in (\ref{hatT0T+T-}),
\begin{eqnarray*}
&& \e^{-\alpha(\widehat S_++\widehat S_-)}\widehat S_0\e^{\alpha(\widehat S_++\widehat S_-)}=\frac{1}{\sqrt{1+2\gamma^2}}(\widehat S_0
+\gamma({\widehat S_+-\widehat S_-}))=\widehat T_0,\\
&&\e^{-\alpha(\widehat S_++\widehat S_-)}\widehat S_+\e^{\alpha (\widehat S_++\widehat S_-)}=\\&&=\frac{1}{\sqrt{1+2\gamma^2}}
\left(\frac{1}{2}~(\widehat S_+-\widehat S_-)+
\frac{1}{2}\sqrt{1+2\gamma^2}(\widehat S_++\widehat S_-)-\gamma{\widehat S_0}\right)=\widehat T_+,\\
&&\e^{-\alpha(\widehat S_++\widehat S_-)}\widehat S_-\e^{\alpha (\widehat S_++\widehat S_-)}=\\&&=\frac{1}{\sqrt{1+2\gamma^2}}
\left(\frac{1}{2}~(\widehat S_--\widehat S_+)+
\frac{1}{2}\sqrt{1+2\gamma^2}(\widehat S_++\widehat S_-)+\gamma{\widehat S_0}\right)=\widehat T_-.
\end{eqnarray*}
\end{pro}
%%%%%%%%%%%%%%%%%%%%%%

Next we observe that, if $\widehat e_j$ is an eigenvector of $\widehat S_0$ associated with the eigenvalue $(j-3/4),$ then the eigenvectors
$\widehat\psi_j$ and $\breve\psi_j$
%are, respectively, eigenvectors
of $\widehat T_0$ and $\widehat T_0^\dag$ associated with the same eigenvalue satisfy
$$
\widehat \psi_j=\e^{-\alpha(\widehat S_++\widehat S_-)}\widehat e_j,~~
\breve \psi_j=\e^{\alpha(\widehat S_++\widehat S_-)}\widehat e_j.
$$
Indeed,
\begin{eqnarray*}
&&\widehat T_0\e^{-\alpha(\widehat S_++\widehat S_-)}\widehat e_j=\e^{-\alpha(\widehat S_++\widehat S_-)}\widehat S_0\widehat e_j=\left(j-\frac{3}{4}\right)\e^{-\alpha(\widehat S_++\widehat S_-)}\widehat e_j,\\
&&\widehat T_0^\dag\e^{\alpha(\widehat S_++\widehat S_-)}\widehat e_j=\e^{\alpha(\widehat S_++\widehat S_-)}\widehat S_0\widehat e_j=\left(j-\frac{3}{4}\right)\e^{\alpha(\widehat S_++\widehat S_-)}\widehat e_j.
\end{eqnarray*}

We observe that $\e^{-\alpha(\widehat S_++\widehat S_-)}$ goes from
${\rm span}\{\widehat e_j\}$ to ${\rm span}\{\widehat \psi_j\}$ and from
${\rm span}\{\breve \psi_j\}$ to ${\rm span}\{\widehat e_j\}$,
while $\e^{\alpha(\widehat S_++\widehat S_-)}$ goes from
${\rm span}\{\widehat \psi_j\}$ to ${\rm span}\{\widehat e_j\}$ and from
${\rm span}\{\widehat e_j\}$ to ${\rm span}\{\breve \psi_j\},$
according with
$$\e^{-\alpha\left(\widehat S_++\widehat S_-\right)}\left({\rm span}\{\widehat e_j\}\right)={\rm span}\{\widehat \psi_j\},
\quad \e^{-\alpha(\widehat S_++\widehat S_-)}\left({\rm span}\{\breve \psi_j\}\right)={\rm span}\{\widehat e_j\},$$
and
$$\e^{\alpha(\widehat S_++\widehat S_-)}\left({\rm span}\{\widehat \psi_j\}\right)={\rm span}\{\widehat e_j\},
\quad \e^{\alpha(\widehat S_++\widehat S_-)}\left({\rm span}\{\widehat e_j\}\right)={\rm span}\{\breve \psi_j\}.
%\e^{\alpha(\widehat S_++\widehat S_-)}\left(\left({\rm span}\{\widehat e_j\}\right)\right)={\rm span}\{\breve \psi_j\}.
$$

It follows that,
$$\widehat \psi_j=\e^{-2\alpha(\widehat S_++\widehat S_-)}\breve\psi_j$$
and so,
$$\e^{2\alpha(\widehat S_++\widehat S_-)}\left({\rm span}\{\widehat \psi_j\}\right)={\rm span}\{\breve \psi_j\},
\quad \e^{-2\alpha(\widehat S_++\widehat S_-)}\left({\rm span}\{\breve \psi_j\}\right)={\rm span}\{\widehat \psi_j\}.$$
As  consequence,
$$\e^{2\alpha (\widehat S_++\widehat S_-)}\widehat H\e^{-2\alpha (\widehat S_++\widehat S_-)}=\widehat H^\dag,$$
and we find the {\it metric matrix}
$$\widehat D^2=\e^{2\alpha(\widehat S_++\widehat S_-)},$$
which induces the inner product
$$\langle\cdot,\cdot\rangle_{\widehat D^2}:=\langle\widehat D^2\cdot,\cdot\rangle.$$

The following result is now easily obtained.
\begin{pro}
With respect to the inner product $ \langle\cdot,\cdot\rangle_{\widehat D^2},$ the infinite matrix $\widehat H_0$ is Hermitian.
Moreover, $\widehat T_0$ is also Hermitian and $\widehat T^+$ is the adjoint of $\widehat T_-.$
\end{pro}
\begin{proof}
The first assertion is proved noting that
$$
\langle\widehat H \widehat\psi,\widehat\psi \rangle_{\widehat D^2}=
\langle\widehat D^2\widehat H \widehat\psi,\widehat\psi \rangle=
\langle \widehat\psi,\widehat H^\dag\widehat D^2\widehat\psi \rangle=
\langle \widehat\psi,\widehat D^2\widehat H\widehat\psi \rangle=
\langle\widehat D^2 \widehat\psi,\widehat H\widehat\psi \rangle=
\langle\widehat\psi,\widehat H \widehat\psi \rangle_{\widehat D^2}.
$$
The next assertion is similarly shown.
The last assertion clearly follows, observing that
$$\widehat D^{-1}\widehat S_+\widehat D=\widehat T_+,
~~\widehat D^{-1}\widehat S_-\widehat D=\widehat T_-,$$
implies
$$~\widehat D^2\widehat T_+=\widehat T_-^\dag\widehat D^2.$$
\end{proof}
\section{A fermionic model}\label{S3}
%%%%%%%%%%%%%%%%%%%%%%%%%%%%%%%%%%%/
%Come\c{c}ar pelo conceito de espa\c{c}o $\C^\infty$ e do espa\c{c}o $\cal H$ dos tensores completamente antissim\'etricos
%sobre $\C^\infty$?

We are concerned with the following non-Hermitian Hamiltonian
\begin{eqnarray*}
&&H=\frac{1}{4}c_1^\dag c_1+\frac{5}{4} c_2^\dag c_2+ \frac{9}{4}c_3^\dag c_3+\ldots\\&&+\gamma\left(
\sqrt{\frac{1\cdot2}{8}}(c_2^\dag c_1-c_1^\dag c_2)
+\sqrt{\frac{3\cdot4}{8}}(c_3^\dag c_2-c_2^\dag c_3)
+\sqrt{\frac{5\cdot6}{8}}(c_4^\dag c_3-c_3^\dag c_4)+\ldots\right),~~ \gamma\in\R,
\end{eqnarray*}
where $c_j^\dag$ and its adjoint $c_j$ are fermionic operators.
%%%%%%%%%%%%%%%%%%%%%%%%
%that is, operators satisfying the
%anticommutation relations
%%for the fermionic operators $c_j^\dag,c_j$ are
%$$\{c_i^\dag, c_j\}=c^\dag_i c_j+c_jc^\dag_i=\delta_{ij},~~\{c_i^\dag, c^\dag_j\}=\{c_i, c_j\}=0,~~i,j=1,2,3,\ldots,$$
%where $\delta_{ij}$ is the {\it Kronecker symbol}. Notice that
%$$[c^\dag_ic_j,c^\dag_k]=-\delta_{jk}c^\dag_i, ~~[c^\dag_ic_j,c_k]=\delta_{ik}c_j.$$
%%%%%%%%%%%%%%%%%%%%%

In terms of the Hermitian operators
\begin{eqnarray*}
&&S_0= \frac{1}{4}c_1^\dag c_1+\frac{5}{4}c_2^\dag c_2+\frac{9}{4}c_3^\dag c_3+\ldots\\&&
S_-=
\sqrt{\frac{1\cdot2}{8}} c_1^\dag c_{2}+\sqrt{\frac{3\cdot4}{8}} c_2^\dag c_{3}+\sqrt{\frac{5\cdot6}{8}} c_3^\dag c_{4}+\ldots\\&&
S_+=
\sqrt{\frac{1\cdot2}{8}} c_2^\dag c_{1}+\sqrt{\frac{3\cdot4}{8}} c_3^\dag c_{2}+\sqrt{\frac{5\cdot6}{8}} c_4^\dag c_{3}+\ldots,
\end{eqnarray*}
the Hamiltonian $H$ is expressed as
$$H=S_0+\gamma(S_+-S_-).$$

The following commutation relations are easily seen to hold
\begin{eqnarray*}
[S_-,S_0]=S_-,\quad
[S_0,S_+]=S_+,\quad
[S_-,S_+]=S_0.\quad
\end{eqnarray*}
%%%%%%%%%%%%%%%%%%%%%%
\subsection{%Eigenvalues and eigenvectors of $H$
Raising and lowering operators}
Since $H,S_0,S_+,S_-$ commute with $N_{op}$, the eigenspaces of $N_{op}$, namely 
${\cal A}_0,~$${\cal A}_1,~$${\cal A}_2,~,\ldots$
are {\it invariant spaces} of $H,S_0,S_+,S_-$.
We also notice that $H$, $H^\dag$ and $N_{op}$ have the same {\it vacuum} $\phi_0\in{\cal A}_0$:
$$H\phi_0=H^\dag\phi_0=N_{op}\phi_0=0.$$

In order to obtain the eigenvalues and eigenvectors of $H$ the equation of motion method (EMM) is also used,
\begin{eqnarray*}
&&[H,z{S_0}+x{S_+}+y{S_-}]={S_0}(-\gamma x-\gamma y)+{S_+}(-\gamma z+x)+{S_-}(-\gamma z- y)\\
&&=\Lambda(z{S_0}+x{S_+}+y{S_-}),~~\Lambda,x,y,z\in\R. %,~~k=1,2,3,\ldots.
\end{eqnarray*}
This method leads to the $3\times3$ matrix eigenproblem,
$$
\left[\begin{matrix}0&-\gamma&-\gamma\\
-\gamma&1&0\\
-\gamma&0&-1\end{matrix}\right]
\left[\begin{matrix}z\\x\\y\end{matrix}\right]=\Lambda
\left[\begin{matrix}z\\x\\y\end{matrix}\right],~~\Lambda\in\R,
$$
whose eigenvalues are readily obtained,
$$\Lambda_0=0,~\Lambda_1= - \sqrt{1 + 2\g^2},~\Lambda_2=  \sqrt{1 + 2\g^2},$$
as well as the respective eigenvectors,
\begin{eqnarray*}
&&u_0=%\frac{1}{\sqrt{1 + 2 \g^2}}
\left[(1, \g, - \g)\right]^T,\\
&&u_-=\left[{1},- \frac{1 - \sqrt{1 + 2 \g^2}}{2\g},
\frac{1+\sqrt{1+2\g^2}}{2\g}\right]^T,\\
&&u_+=\left[{1}, - \frac{1 + \sqrt{1 + 2 \g^2}}{2\g},
\frac{1-\sqrt{1+2\g^2}}{2\g}\right]^T.
\end{eqnarray*}
From the normalized eigenvectors the following operators are constructed:
\begin{eqnarray}
&&{T_0}=\frac{1}{\sqrt{1 + 2 \g^2}}({S_0}+ \g ({S_+} - {S_-})),\nonumber\\
&&{T_-}=\frac{\g}{{\sqrt{1 + 2 \g^2}}}{S_0}- \frac{1 - \sqrt{1 + 2 \g^2}}{2 {\sqrt{1 + 2 \g^2}}}{S_+}
 +\frac{1+\sqrt{1+2\g^2}}{2{\sqrt{1 + 2 \g^2}}}{S_-},\nonumber\\
 &&{T_+}=-\frac{\g}{{\sqrt{1 + 2 \g^2}}}{S_0} + \frac{1 + \sqrt{1 + 2 \g^2}}{2 {\sqrt{1 + 2 \g^2}}}{S_+}
  -\frac{1-\sqrt{1+2\g^2}}{2{\sqrt{1 + 2 \g^2}}}{S_-}.\label{T0T+T-}
\end{eqnarray}
These operators
 obey the same commutation relations as the operators $S_0,S_+,S_-$  which characterize the $su(1,1)$ algebra,
\begin{eqnarray*}
[T_-,T_0]=T_-,\quad
[T_0,T_+]=T_+,\quad
[T_-,T_+]=T_0.\quad
\end{eqnarray*}
%%%%%%%%%%%%%%%%%%%%%%%%%%%%%%%%%%%%
We say that $T_+$ is a {\it raising operator}, because, if $\lambda$ is an eigenvalue of $H$ with
eigenvector $\psi$, i.e.,
$$H\psi=\lambda\psi,$$
then $(\lambda+\sqrt{1+2\gamma^2})$ is an eigenvalue of $H$ with
eigenvector $T_+\psi$, that is,
$$HT_+\psi=\left(\lambda+\sqrt{1+2\gamma^2}\right)T_+\psi.$$
Similarly, $T_-$ is a  {\it lowering operator}, because
%if $\lambda$ is an eigenvalue of $H$ with eigenvector $\psi$, i.e.,$$H\psi=\lambda\psi,$$then
$(\lambda-\sqrt{1+2\gamma^2})$ is an eigenvalue of $H$ with
eigenvector $T_-\psi$, that is,
$$HT_-\psi=\left(\lambda-\sqrt{1+2\gamma^2}\right)T_-\psi,$$
provided $T_-\psi\neq0.$
However, %the {\it raising operator}
$T_+$ is not the adjoint of
$T_-,~T_-\neq(T_+)^\dag$ and $T_0$ is not Hermitian, $T_0\neq T_0^\dag$. Due to these facts, we say that the operators $T_0,T_+,T_-$ generate a {\it pseudo}-$su(1,1)$ algebra.
%%%%%%%%%%%%%%%%%

We have shown the following.

\begin{pro}\label{P4.1}
The eigenvalues of $H$ associated with eigenvectors in ${\cal A}_1$ are $\sqrt{1+2\gamma^2}$ $(1/4,5/4,9/4,\ldots).$
The eigenvectors of $H$ in ${\cal A}_1$ are
$$\psi_n={T_+}^{n-1}\psi_1,~~n=2,~3,~4,~\ldots,$$
where $\psi_1\in{\cal A}_1$ is such that
$${T_-}{\psi}_1=0.$$
\end{pro}
\begin{proof}
The result follows, observing that
the eigenvalues of $T_0$ associated with eigenvectors in ${\cal A}_1$ are $1/4,5/4,9/4,\ldots.$
\end{proof}
\begin{pro}\label{P4.2}
The eigenvectors of ${H^\dag}$ in ${\cal A}_1$ are
$${{\widetilde\psi}}_n=({T_-}^\dag)^{n-1}{\widetilde\psi}_1,$$
where ${{\widetilde\psi}}_1\in{\cal A}_1$ satisfies
$$({T_+})^\dag{{\widetilde\psi}}_1=0.$$
\end{pro}

%The eigenvectors ${\psi}_1,~{\widetilde\psi}_1$ are easily obtained.
The eigenvector systems $\{\psi_n\},~\{\widetilde\psi_n\}$
may be made biorthonormal so that
$$\langle\widetilde\psi_m,\psi_n\rangle=\delta_{mn}\langle\widetilde\psi_n,\psi_n\rangle.$$
%%%%%%%%%%%%%%%%%%%%%%%%

We observe that the restriction of the operator $H$ to ${\cal A}_1$ is identified with the matrix $\widehat H$
acting on $\C.$
%%%%%%%%%%%%%%%%%%%%%%%%%%
\section{Dynamical fermionic operators}\label{S4}
The {\it dynamical fermionic operators} are linear combinations $x_1c_1^\dag+x_2c_2^\dag+x_3c_3^\dag+\ldots,~x_1,x_2,x_3,\ldots\in\R$, such that
\begin{eqnarray*}&&\left[H,(x_1c_1^\dag+x_2c_2^\dag+x_3c_3^\dag+\ldots)\right]=\left(\frac{1}{4}x_1-\sqrt{\frac{1\cdot2}{8}}\gamma x_2\right)c_1^\dag
\\&&+\left(\sqrt{\frac{1\cdot2}{8}}\gamma x_1+\frac{5}{4}x_2-\sqrt{\frac{3\cdot4}{8}}x_3\right)c_2^\dag
%\\&&
+\left(\sqrt{\frac{3\cdot4}{8}}\gamma x_2+\frac{9}{4}x_3-\sqrt{\frac{5\cdot6}{8}}x_4\right)c_3^\dag+\ldots
\\&&=\lambda(x_1c_1^\dag+x_2c_2^\dag+x_3c_3^\dag+\ldots),~~\lambda\in\R.
 \end{eqnarray*}
The EMM leads to the eigenproblem,
\begin{eqnarray*}
\frac{1}{2\sqrt2}\left[\begin{matrix}
1/\sqrt2&-\sqrt{{1\cdot2}{}}\gamma&0&0&\ldots\\
\sqrt{{1\cdot2}{}}\gamma&5/\sqrt2&-\sqrt{{3\cdot4}{}}\gamma&0&\ldots\\
0&\sqrt{{3\cdot4}{}}\gamma&9/\sqrt2&-\sqrt{{5\cdot6}{}}\gamma&\ldots\\
0&0&\sqrt{{5\cdot6}{}}\gamma&13/\sqrt2%&-\sqrt{\frac{3\cdot4}{2}}\gamma
&\ldots\\
\vdots&\vdots&\vdots&\vdots&\ddots\end{matrix}\right]~
%\left[\begin{matrix}
%%1&-\sqrt{\frac{1\cdot2}{2}}\gamma&0&0&\ldots\\
%\sqrt{\frac{1\cdot2}{2}}\gamma&2&-\sqrt{\frac{2\cdot3}{2}}\gamma&0&\ldots\\
%0&\sqrt{\frac{2\cdot3}{2}}\gamma&3&-\sqrt{\frac{3\cdot4}{2}}\gamma&\ldots\\
%0&0&\sqrt{\frac{3\cdot4}{2}}\gamma&4%&-\sqrt{\frac{3\cdot4}{2}}\gamma
%&\ldots\\
%\vdots&\vdots&\vdots&\vdots&\ddots\end{matrix}\right]~
\left[\begin{matrix}x_1\\x_2\\x_3\\x_4\\\vdots
\end{matrix}\right]=\lambda
\left[\begin{matrix}x_1\\x_2\\x_3\\x_4\\\vdots
\end{matrix}\right],
\end{eqnarray*}
involving the matrix $\widehat H$ in (\ref{widehatH}),
%which has beenconsidered in 
Section 2.
%as well as the respective spectral analysis.
%%%%%%%%%%%%%%%%%%%%%%%%%%%%%%%%%%%
%%%%%%%%%%%%%%%%%%%%%

%bbbbbbbbbbb
%%%%%%%%%%%%%%%%%%%%%%%%%%%%%%%%%%%%%%%%
Let us express the eigenvector $\widehat\psi_n$ of $\widehat H$ as
$$\widehat\psi_n=\left[x^{(n)}_1,~x^{(n)}_2,~x^{(n)}_3,~\ldots \right]^T,$$
and the eigenvector $\breve\psi_n$ of $\widehat H^T$ as
$$\breve\psi_n=\left[y^{(n)}_1,~y^{(n)}_2,~y^{(n)}_3,~\ldots\right]^T.$$
%These eigenvectors constitute biorthonormal bases,
%$$\langle u_i,v_j\rangle=\delta_{ij}.$$
Pseudo-fermionic operators may now be constructed
$$d_i^\ddag:=x^{(i)}_1c_1^\dag +x^{(i)}_2c_2^\dag+ x^{(i)}_3c_3^\dag+\ldots,$$
$$d_i:=y^{(i)}_1c_1 +y^{(i)}_2c_2 +y^{(i)}_3c_3+\ldots.$$
The following anticommutation relations hold
$$\{d_i^\ddag, d_j\}=d^\ddag_i d_j+d_jd^\ddag_i=\delta_{ij},~~\{d_i^\ddag, d^\ddag_j\}=\{d_i, d_j\}=0,~~i,j=1,2,3,\ldots,n.$$
These operators are called {\it pseudo-fermionic} because $d_i^\dag\neq d_i^\ddag$.

The proof of the next result is independent from the proofs of Propositions \ref{P4.1} and \ref{P4.2}.
\begin{teor}
In terms of the pseudo-fermionic operators,
the Hamiltonian $H$ may be expressed as
$$H=\sqrt{1+2\g^2}\left(\frac{1}{4}d_1^\ddag d_1+\frac{5}{4} d_2^\ddag d_2+\frac{9}{4} d_3^\ddag d_3+\ldots\right).$$
Further,
$$\sigma(H)=\left\{\sqrt{1+2\gamma^2}\sum_{k=1}^\infty \frac{4k-3}{4}n_k,~~n_k\in\{0,1\},~~k=1,2,3,\ldots\right\},$$
and the associated eigenvectors are expressed as
$$\Psi_{n_1,n_2,n_3,\ldots}=\left((d_1^\ddag)^{n_1}(d_2^\ddag)^{n_2}(d_3^\ddag)^{n_3}\ldots\right)\psi,~~\psi\in{\cal A}_0
,~~n_k\in\{0,1\},~~k=1,2,3,\ldots.
$$
\end{teor}
\begin{proof}
Denote by $\lambda_n$ the common eigenvalue of $\widehat H$ and $\widehat H^T$ associated, respectively
with the eigenvectors $\widehat\psi_n$ and $\breve\psi_n,$
$$
\widehat H\widehat\psi_n=\lambda_n\widehat\psi_n,~~~
\widehat H\breve\psi_n=\lambda_n\breve\psi_n.
$$

Let $\widehat U$ and $\breve U$ be the matrices %defined as
$$
\widehat U=[\widehat\psi_1,~\widehat\psi_2,~\widehat\psi_3,~\ldots],~~~
\breve U=[\breve\psi_1,~\breve\psi_2,~\breve\psi_3,~\ldots],
$$
whose columns are the eigenvectors $\widehat\psi_n$ and $\breve\psi_n$
of $\widehat H$ and $\widehat H^\dag$
and let $$\widehat H_{diag}={\rm diag}(\lambda_1,~\lambda_1,~\lambda_3,~\ldots),$$
(cf. Propositions \ref{P2.1} and \ref{P2.2}).
Then, we may write
$$
\widehat H\widehat U=\widehat U\widehat H_{diag},~~~
\widehat H\breve U=\breve U\widehat H_{diag}.
$$
Moreover from the biorthonormality of the eigenvectors we get,
$$\widehat U\breve U^T=\breve U^T\widehat U=I.$$
Notice that
\begin{eqnarray*}
H=[c_1^\dag,c_2^\dag,c_3^\dag,\ldots]\widehat H[c_1,c_2,c_3,\ldots]^T.
\end{eqnarray*}
Indeed,
\begin{eqnarray*}
&&[d_1^\ddag,d_2^\ddag,d_3^\ddag,\ldots]\widehat H_{diag}[d_1,d_2,c_3,\ldots]^T\\
&&=[d_1^\ddag,d_2^\ddag,d_3^\ddag,\ldots]\breve U^T \widehat H_{diag}\widehat U[d_1,d_2,c_3,\ldots]^T\\
&&=[c_1^\dag,c_2^\dag,c_3^\dag,\ldots]
(\widehat U\breve\Phi^T) \widehat H(\widehat U\breve\Phi^T)[c_1,c_2,c_3,\ldots]^T\\
&&=[c_1^\dag,c_2^\dag,c_3^\dag,\ldots]\widehat H[c_1,c_2,c_3,\ldots]^T=H.
%\breve\Phi^T H\widehat \Phi
\end{eqnarray*}
It is now easy to show that
$$H\Psi_{n_1,n_2,\ldots}=(n_1\lambda_1+n_2\lambda_2+\ldots)\Psi_{n_1,n_2,\ldots}.$$
%The remaining is obvious.
\end{proof}\\

\begin{teor}
In terms of the pseudo-fermion operators, $ T_0, T_+, T_-,$
defined in (\ref{T0T+T-}) may be expressed as
\begin{eqnarray*}
&&T_0=\frac{1}{4} d_1^\ddag d_1+\frac{5}{4}d_2^\ddag d_2+\frac{9}{4}d_3^\ddag d_3+\ldots,\\&&
T_-=\sqrt{\frac{1\cdot2}{8}} d_1^\ddag d_{2}+\sqrt{\frac{3\cdot4}{8}} d_2^\ddag c_{3}+\sqrt{\frac{5\cdot6}{8}} d_3^\ddag d_{4}+\ldots,\\&&
T_+=\sqrt{\frac{1\cdot2}{8}} d_2^\ddag d_{1}+\sqrt{\frac{4\cdot4}{8}} d_3^\ddag c_{2}+\sqrt{\frac{5\cdot6}{8}} d_4^\ddag c_{3}+\ldots.
\end{eqnarray*}
\end{teor}
\begin{proof}Analogous to the previous proof.
\end{proof}
\section{Physical Hilbert space}\label{S6}
Let
$${\cal S}={ \rm span}\{\psi_i,\psi_i\wedge\psi_j,\psi_i\wedge\psi_j\wedge\psi_k,\ldots:i<j<k<\ldots;i,j,k,\ldots=1,2,3,\ldots \},$$
$$\widetilde{\cal S}={ \rm span}\{\widetilde\psi_i,\widetilde\psi_i\wedge\widetilde\psi_j,\widetilde\psi_i\wedge\widetilde\psi_j\wedge\widetilde\psi_k,
\ldots:i<j<k<\ldots;i,j,k,\ldots=1,2,3,\ldots \}.$$
We find
\begin{eqnarray*}
&&\langle\psi_i,\widetilde\psi_{i'}\rangle=0,~
\langle\psi_i\wedge\psi_j,\widetilde\psi_{i'}\wedge\widetilde\psi_{j'}\rangle=0,~
\langle\psi_i\wedge\psi_j\wedge\psi_k,\widetilde\psi_{i'}\wedge\widetilde\psi_{j'}\wedge\widetilde\psi_{k'}\rangle=0,~\ldots,\\&&
%i<j<k<\ldots,~
i\neq i',~(i,j)\neq(i',j'),~(i,j,k)\neq(i',j',k'),~\ldots
%i,j,k,\ldots=1,2,3,\ldots,
%~i'<j'<k'<\ldots.
~,i,j,k,\ldots,~i',j',k',\ldots=1,2,3,\ldots.
\end{eqnarray*}
Let us define $\textsf{D}:{\cal S}\rightarrow\widetilde{\cal S}$ such that
\begin{eqnarray*}
&&\textsf{D}\psi_i:=\widetilde\psi_{i},~
\textsf{D}(\psi_i\wedge\psi_j):=\widetilde\psi_{i}\wedge\widetilde\psi_{j},~
%S^\bar+ + S^\bar−——>S^\hat+ + S^\hat−. I don't understand this one.
 \textsf{D}(\psi_i\wedge\psi_j\wedge\psi_k):=\widetilde\psi_{i}\wedge\widetilde\psi_{j}\wedge\widetilde\psi_{k}=0,~\ldots,\\&&
i<j<k<\ldots,
~i,j,k,\ldots=1,2,3,\ldots.
\end{eqnarray*}
For $\Phi,\Psi\in{\cal S}$, we define the {\it inner product}
$$\langle\Phi,\Psi\rangle_\textsf{D}=\langle \textsf{D}\Phi,\Psi\rangle.$$

Following Mostafazadeh, we say that  the {\it physical Hilbert space}
\cite{mostafa,mostafa1} is the set $\cal S$ endowed with the inner product $\langle \cdot,\cdot\rangle_\textsf{D}.$
It may easily be seen that, for $0\neq\Phi\in{\cal S},$ we have $\langle\Phi,\Phi\rangle_\textsf{D}>0.$
%$$T_0=\frac{1}{\sqrt{1+2\gamma^2}}(d_1^\ddag d_1+2 d_2^\ddag d_2+3 d_3^\ddag d_3+\ldots).$$

The {\it physical numerical range} of $A$ is defined as
$$W_{phys}(A)=\left\{\langle\Phi A,\Phi\rangle_\textsf{D}:\langle\Phi,\Phi\rangle_\textsf{D}=1,~\Phi\in{\cal H}\right\}.$$
and is useful in the next section.
\section{Statistical thermodynamics of non-Hermitian Hamiltonians with real eigenvalues}\label{S5}
The main objective of this section is to present the  description, according to statistical thermodynamics, of
a system characterized by a non-Hermitian Hamiltonian possessing real eigenvalues. Conserved quantities are operators
that have real eigenvalues and commute with $H$. Since $H$ is not Hermitian,
the conserved quantities may not be Hermitian, but they have the same eigenvectors as $H$,
and are Hermitian with respect to the  norm induced by $\langle\cdot,\cdot\rangle_\textsf{D}.$
%%%%%%%%%%%%%%%%In this case,

In statistical thermodynamics, pure states are given by vectors and  mixed states are described by  {\it density matrices},
i.e., positive semidefinite Hermitian matrices with trace 1.
Observable quantities are represented by Hermitian matrices.
For a system with a Hermitian Hamiltonian $H$ and {\it fermionic number operator}
%The fermionic number operator is expressed as
$$N_{op}=\sum_{i=1}^\infty c_i^\dag c_i,$$ the density
matrix of the {\it  equilibrium thermal state} is
$$
\rho_{eq}=\frac{\e^{-\beta H-\zeta N_{op}}}{\Tr \e^{-\beta H-\zeta N_{op}}},
$$
where $\beta$ is the inverse of the {\it absolute temperature} $T$ and $\zeta$ is related to the so called
{\it chemical potential} $\mu$ according to $\zeta=-\beta\mu.$

If $H$ is not Hermitian,
also $\rho_{eq}$ is not Hermitian, but it has real eigenvalues.
The density matrix encapsulates the statistical properties of the system.

The {\it partition function} $Z$ \cite{beb,landau} is
\begin{equation}Z=\Tr \exp(-\beta H-\zeta N_{op}).\label{Z}\end{equation}
According to statistical thermodynamics, the equilibrium properties of the system may be derived
from the logarithm of the partition function,
%The partition function is closely related to the
while in classical thermodynamics,  the equilibrium properties of a system may be derived from its
{\it thermodynamical potential} %of the system, defined by
$$F=E-\mu N-TS
%=-\frac{1}{\beta}\log Z
,$$
where $E$ is the internal energy, $\mu$ is the chemical potential, $N$ is the number of particles, understood as
the amount of some chemical compound,
and $S$ the classical {\it entropy}.  In statistical thermodynamics, $E$ becomes the expectation value of $H$,
while $N$ is identified with $\langle N_{0p}\rangle$, the expectation value of $N_{op}$.
Obviously,  $F$  is identified with $-\log Z/\beta,$ since the roles played by
both quantities are parallel,
$$F=-\frac{1}{\beta}\log Z.$$
This identification provides the statistical definition of entropy.
%It is known \cite{landau} that
We note that
$$E=\Tr(\rho_{eq}H)=-\frac{\partial\log Z}{\partial\beta}$$
and
$$N=\Tr(\rho_{eq}N_{op})=-\frac{\partial\log Z}{\partial\zeta}$$
are, respectively,
the {\it expected values} of $H$ and of $N_{op}$ at statistical equilibrium,
that is, the equilibrium expectation values of the respective physical measurements.
%$$E=-\frac{\partial}{\partial\beta}\log Z,~N=\frac{\partial}{\partial\zeta}\log Z.$$

By the following computation
\begin{eqnarray*}
&&-\Tr(\rho_{eq}\log\rho_{eq})\\
&&=\Tr\left(\frac{\e^{-\beta H-\zeta N_{op}}}{\Tr(\e^{-\beta H-\zeta N_{op}})}\times(\beta H+\zeta N_{op}+\log  Z)\right)\\
&&=\beta E+\zeta\langle N\rangle+\log Z\\
&&=\beta (E-\mu\langle N\rangle+\frac{1}{\beta}\log Z)\\
&&=\beta (E-\mu\langle N\rangle+(TS-E+\mu\langle N\rangle))=S,
\end{eqnarray*}
we get
%Further, it may be  shown that the entropy  is
%$T=1/\beta$ is the {\it absolute temperature},  $\mu=\zeta/\beta$ is the chemical potentialand
$$S=-\Tr(\rho_{eq}\log\rho_{eq}).$$
%(A entropia de von Neumann \'e $-\Tr(\rho\log\rho)$ quando $\rho\neq\rho_{eq}.$)
Recall that the {\it von Neumann entropy} is given by  $-\Tr(\rho\log\rho)$ for an arbitrary $\rho$, $\rho\neq\rho_{eq}$.

We observe that, if $\lambda_k$ are the eigenvalues of the matrix $\widehat H$, then the eigenvalues
of $\widehat H-\mu I$ are $\lambda_k-\mu.$
%We propose that the role
Thus,  the number operator $N_{op}=\sum_{k=1}^\infty c^\dag_k c_k$ may be replaced, in the expression of $Z$,
 by the pseudo-fermionic number operator
\begin{equation} N_{op}=\sum_{k=1}^\infty d^\ddag_k d_k.\label{Nop}\end{equation}
This is in consonance with the corresponding expression for $H$,
\begin{equation}H=\sum_{k=1}^\infty\lambda_kd^\ddag_k d_k.\label{Hdddagd}\end{equation}
In the definition of the partition function, (\ref{Nop}) and (\ref{Hdddagd}) shall be used.
This ensures that $Z$ is real and positive even though $\lambda_kd^\ddag_k d_k$ and $d^\ddag_k d_k$ are non-Hermitian.

Next, we obtain approximations
for $E$ and $N$, which are valid if the temperature is sufficiently high.
\subsection{A numerical example}\label{S7}
%\begin{figure}[ht]
%\centering
%\includegraphics[width=.5\textwidth, height=0.25\textwidth]
%{antitrid.eps} \caption{Boundary generating curves of the matrix $D_8$. There appear 4 ellipsis, one of them degenerates into a line segment.
%There are 8 boundary generating curves. Each ellipse is the superposition of 2 ellipsis.}
%\label{figD}
%\end{figure}
\begin{figure}[ht]
\centering
\includegraphics[width=.75\textwidth, height=0.75\textwidth]
{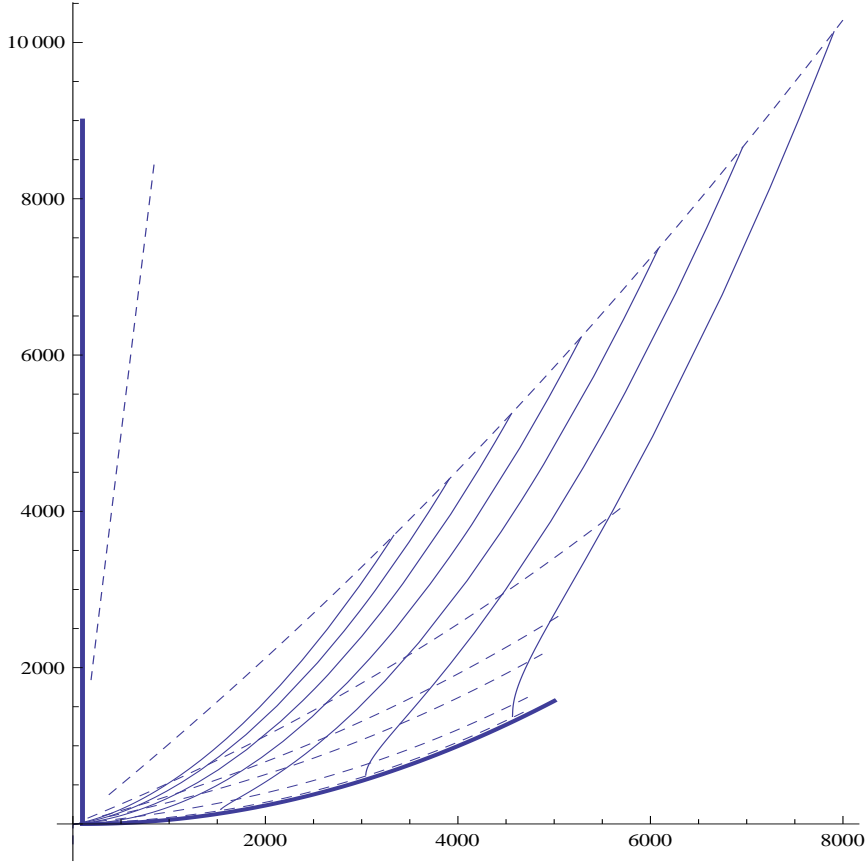} \caption{We have considered $\gamma=3/5$ . Dashed lines represent $\langle H\rangle$ vs $\langle N_{op}\rangle$ for  variable $\mu$
with fixed values of $\beta$.
%with $\beta=0.001,~0.01,~0.02,$ $~0.03,~0.04,~0.08,~0.2$,
%from top to bottom.
%%%%%%%%%%%%%%%
%Except for $\beta=0.001,$ we have taken
% $-15\sqrt{1+2\gamma^2}<\mu<15\sqrt{1+2\gamma^2}$.
Full lines,  represent $\langle H\rangle$ vs $\langle N_{op}\rangle$
for  variable $\beta$ with fixed values of $\mu$.
%with $\mu=\sqrt{1+2\gamma^2}~(-14.75,~-9.75,~-4.75,~\ldots,
%~10.25,~15.25)$,
%from left to right.
The crossing points define pairs $(\langle E\rangle,\langle N_{op}\rangle)$ corresponding to pairs
$(\beta,\mu)$. The horizontal scale should be divided by 100.
The thick line represents the boundary of $W_{phys}(H+iN_{op}),$ the physical numerical range of $H+iN_{op}$.
%The curve for $\beta=0.001$, with $-6000<\mu<-4500$, was included to illustrate
%that points near the boundary of  $W_{phys}(H+iN_{op}),$ are described for very low $\beta.$
}
\label{fig9k}
\end{figure}
%The obtained approximations are applied  to our model.
%By the technique in the previous section, we find that t
The eigenvalues of $\beta H+\zeta N_{op}$ are
$$(\beta\sqrt{1+2\g^2}~(k-3/4)+\zeta)n_k,~k=1,2,3,\ldots,~n_k=0,1.$$
Thus
$$Z=\sum_{n_1,n_2,n_3,\ldots\in\{0,1\}}\exp\left({- \sum_k^\infty(\beta \sqrt{1+2\g^2}~ (k-3/4) +\zeta)n_k}\right),$$
so that,
$$\log Z=\sum_{k=1}^\infty\log\left(1+\exp({-\beta\sqrt{1+2\g^2}~k-\zeta'})\right),$$
where $\zeta'=\zeta-\frac{3}{4}\sqrt{1+2\gamma^2}\beta.$
%%%%%%%%%%%%%%%%%%%%%
%%%%%%%%%%%%%%%%%%%%%%

The {\it Euler-Maclaurin (E-M) formula}  is an important tool in numerical analysis.
%and one of the most remarkable formulas in mathematics.
It
estimates a sum $\sum_{k=0}^ng(k)$ through the integral $\int_0^ng(t)\d t$ with an error term involving
{\it Bernoulli numbers} and {\it polynomials} \cite{spivey}.
\begin{lema} Let $g(t)$ be a real function of class $C^2$. Then,
%One form of the E-M formula states
\begin{eqnarray}
&&
\sum_{k=0}^{n-1}g(k)=
\int_0^ng(t)\d t-{1\over2}(g(n)-g(0))\nonumber\\&&\label{EMacl}+{1\over12}(g'(n)-g'(0))
-{1\over2}\int_0^nB_2(\{t\}){g''(t)}\d t,
\end{eqnarray}
where $k$ is a non negative integer, $B_2(x)=x^2-x+1/6$ is the
{\it second Bernoulli polynomial} and $\{t\}$ denotes the fractional part of $t$.
%The condition imposed on the real function
%$g$ is that it should have continuous second derivative for $t\in(0,n)$.
\end{lema}

The {\it polylogarithm} is the function defined by the power series
$$Li_s(z)=\sum_{k=1}^\infty\frac{z^k}{k^s},~~s\in\C.$$
The case $s=2$ will be used in the next proposition.
\begin{pro}\label{P6}
Let $E$ and $\langle N_{op}\rangle$ be the expectation values of $H$ and $N_{op}$, respectively.
% the expectation value of $N_{op}$.
Then,
\begin{eqnarray*}
%-((E^ze ka)/(12 (1 + E^ze))) - PolyLog[2, -E^ze]/(be^2 ka)
&&E%=-\frac{\partial}{\partial \beta}\log Z
=\frac{1}{\beta^2\sqrt{1+2\g^2}}{\rm Li}_2(-\e^{\zeta'})
-\frac{ \sqrt{1+2\gamma^2}}{12 (1 + \e^{\zeta'})} +\ldots-\frac{3}{4}\sqrt{1+2\gamma^2}\langle N_{op}\rangle,\\
%E^ze/(2 (1 + E^ze)) - (be E^(2 ze) ka)/(12 (1 + E^ze)^2) + (be E^ze ka)/(12 (1 + E^ze)) + Log[1 + E^ze]/(be ka)
&&\langle N_{op}\rangle=
%\frac{\partial}{\partial\zeta}\log Z=
\frac{1}{\beta\sqrt{1+2\g^2}}\log(1+\e^{-\zeta'})
+\frac{1}{2 (1 + \e^{\zeta'})} - \frac{\beta  \sqrt{1+2\gamma^2}}{12 (1 + \e^{\zeta'})^2}+\ldots.
\end{eqnarray*}
\end{pro}
\begin{proof}
We observe that
$$E=-\frac{\partial}{\partial \beta}\log  Z,~~~\langle N_{op}\rangle=\frac{\partial}{\partial\zeta}\log Z.$$
By the Euler-Maclaurin formula and using {\it Mathematica},
\begin{eqnarray*}
&&\log Z=\int_{0}^\infty\d x\log\left(1+\e^{-\beta\sqrt{1+2\g^2}~x-\zeta'}\right)+\frac{1}{2}\log({1+\e^{-\zeta'}})
+\frac{\beta\sqrt{1+2\gamma^2}\e^{-\zeta'}}{12(1+\e^{-\zeta'})}+\ldots\\
&&=-\frac{1}{\beta\sqrt{1+2\g^2}}{\rm Li}_2(-\e^{\zeta'})+\frac{1}{2}\log({1+\e^{-\zeta'}})
+\frac{\beta\sqrt{1+2\gamma^2}\e^{-\zeta'}}{12(1+\e^{-\zeta'})}+\ldots.
\end{eqnarray*}

Thus, the result follows.
\end{proof}

%However, the first term diverges for $\beta=0$, so that for very small $\beta$
%the expansion should be used with caution, due to numerical  problems.
%Nevertheless, the first term provides a useful upper bound to $\log Z$.
% and so is meaningful and useful.
%where ${\rm Li}_2(z)$ is the Polylogarithm function. Thus,

In Fig. \ref{fig9k}, dashed lines represent $\langle H\rangle$ vs $\langle N_{op}\rangle$ for $\gamma=3/5$ and variable $\mu$
with $\beta=0.001,~0.01,~0.02,$ $~0.03,~0.04,~0.08,~0.2$,
from top to bottom.
Except for $\beta=0.001,$ we have taken
 $-15\sqrt{1+2\gamma^2}<\mu<15\sqrt{1+2\gamma^2}$.
Full lines,  represent $\langle H\rangle$ vs $\langle N_{op}\rangle$
for  variable $\beta$ with $\mu=\sqrt{1+2\gamma^2}~(-14.75,~-9.75,~-4.75,~\ldots,
~10.25,~15.25)$,
from left to right. The crossing points define pairs $(\langle E\rangle,\langle N_{op}\rangle)$ corresponding to pairs
$(\beta,\mu)$. The horizontal scale should be divided by 100.
The thick line represents the boundary of $W_{phys}(H+iN_{op}),$ the physical numerical range of $H+iN_{op}$
(the vertical line parallel to the $y$ axis and the lower
parabolic arc).
The curve for $\beta=0.001$, with $-6000<\mu<-4500$, was included to illustrate
that points near the boundary of  $W_{phys}(H+iN_{op}),$ are described for very low $\beta.$

We remark that in the present case the function $g(x)$ is of class $C^\infty$, so that the Euler-Maclaurin formula leads to an
expansion in powers of $\beta\sqrt{1+2\gamma^2}$ for $\log Z$ which
may be carried out indefinitely. However, in Proposition \ref{P6}, only the first three terms of this expansion
have been considered.
%Some of the numerically obtained pairs $\langle H\rangle,\langle N_{op}\rangle$ lie outside $W_{pys}(H+iN_{op})$.
%This is a clear indication that the approximation involved in Proposition \ref{P6} was not adequate for those points,
%which should be discarded.
The results are very good for points of $W_{pys}(H+iN{op})$ which are not close to its boundary. For closer points, the full expansion
may be needed.
%In Fig. \ref{fig9k},~ for $\gamma=3/5$, the number operator expectation value
%$N$ is represented versus the energy $E$ by dashed lines,  for
%$\beta=0.2,$ $0.3,$ $0.6,$ $0.8$,
%and by full lines,
%for $\zeta=-3/2,$ $-1,~-1/2,~0,~1/2,~1$.
%The crossing points define pairs $(E,N)$ corresponding to pairs $(\beta,\eta)$.
%%%%%%%%%%%%%%%%%%%%%%%%%%%%%%%%%%%%%%%%%%%%%%%%%%%%
\section{Conclusions}\label{S8}
%%%%%%%%%%%%%%%%%%%%%%%%%
%The paper is organized as follows. In Section \ref{S2}, t
We have investigated the spectrum  of a non-Hermitian semi-infinite matrix
$\widehat H$, %with real eigenvalues}is performed.In Section \ref{S2a},
and we have explicitly constructed a metric matrix which renders $\widehat H$ Hermitian.
% is constructed. The obtained results are crucial in the remaining parts ofthe paper. In Section \ref{S3},
 {A fermionic model} characterized by {a non-Hermitian Hamiltonian with real eigenvalues}
has been investigated. %In Section \ref{S4},
{Dynamical pseudo-fermionic operators} have been constructed
in terms of which the fermionic Hamiltonian acquires diagonal form.
A physical Hilbert space, allowing for the probabilistic interpretation of the model according to quantum mechanics,
has been introduced.  Approximate expressions for the energy expectation value
and the number operator expectation value, in terms of the absolute temperature $T$ and of the chemical potential $\mu$, are obtained, based
on the  Euler-Maclaurin formula.
% The thermodynamical inequality for such systems
%using the results in Section \ref{S2}.
%In Section \ref{S5},
Statistical thermodynamics  considerations, in which the physical Hilbert space plays an important role, are applied to the studied fermionic Hamiltonian.
\section*{Acknowledgments}
This work was
partially supported by Funda\c{c}\~ao para a Ci\^encia e Tecnologia, Portugal, under the Project UID/FIS/04564/2019,
and by
%This work was partially supported by
the Centre for Mathematics of
the University of Coimbra, under the Project UID/MAT/00324/2013, funded by the
Portuguese Government through FCT/MEC and co-funded by the European
Regional Development Fund through the Partnership Agreement PT2020.
%%%%%%%%%%%%%%%%%%%%%%%%%%%%%%%%%%%%%%%%%%%%%%%%

\end{document}